\documentclass[twocolumn,superscriptaddress,
showpacs,preprintnumbers,amsmath,amssymb]{revtex4}

\usepackage{amssymb,amsmath,amsthm}
\usepackage{epsfig}
\newtheorem{Thm}{Theorem}

\newtheorem{Lem}[Thm]{Lemma}

\theoremstyle{definition}

\newtheorem{Exam}[Thm]{Example}

\newcommand{\bra}[1]{{\left\langle #1 \right|}}
\newcommand{\ket}[1]{{\left| #1 \right\rangle}}
\newcommand{\inn}[2]{{\left\langle #1 | #2 \right\rangle}}

\newcommand{\T}{\mbox{$\mathrm{tr}$}}

\begin{document}
\title{Monogamy equality in $2\otimes 2 \otimes d$ quantum systems}

\author{Dong Pyo Chi}
\affiliation{
 Department of Mathematical Sciences,
 Seoul National University, Seoul 151-742, Korea
}
\author{Jeong Woon Choi}
\affiliation{
 Department of Mathematical Sciences,
 Seoul National University, Seoul 151-742, Korea
}
\author{Kabgyun Jeong}
\affiliation{
 Nano Systems Institute (NSI-NCRC),
 Seoul National University, Seoul 151-742, Korea
}
\author{Jeong San Kim}
\affiliation{
 Institute for Quantum Information Science,
 University of Calgary, Alberta T2N 1N4, Canada
}
\author{Taewan Kim}
\affiliation{
 Department of Mathematical Sciences,
 Seoul National University, Seoul 151-742, Korea
}
\author{Soojoon Lee}
\affiliation{
 Department of Mathematics and Research Institute for Basic Sciences,
 Kyung Hee University, Seoul 130-701, Korea
}

\date{\today}

\begin{abstract}
There is an interesting property about multipartite entanglement,
called the monogamy of entanglement.
The property can be shown by the monogamy inequality,
called the Coffman-Kundu-Wootters inequality~[Phys. Rev. A {\bf 61}, 052306 (2000);
Phys. Rev. Lett. {\bf 96}, 220503 (2006)],
and
more explicitly
by the monogamy equality
in terms of the concurrence
and the concurrence of assistance,
$\mathcal{C}_{A(BC)}^2=\mathcal{C}_{AB}^2+(\mathcal{C}_{AC}^a)^2$,
in the three-qubit system.
In this paper,
we consider the monogamy equality
in $2\otimes 2 \otimes d$ quantum systems.
We show that $\mathcal{C}_{A(BC)}=\mathcal{C}_{AB}$ if and only if $\mathcal{C}_{AC}^a=0$,
and also show that
if $\mathcal{C}_{A(BC)}=\mathcal{C}_{AC}^a$ then $\mathcal{C}_{AB}=0$,
while there exists a state in a $2\otimes 2 \otimes d$ system such that
$\mathcal{C}_{AB}=0$ but $\mathcal{C}_{A(BC)}>\mathcal{C}_{AC}^a$.
\end{abstract}

\pacs{
03.65.Ud, 
03.67.Mn  
}
\maketitle

Entanglement provides us with a lot of useful applications
in quantum communications,
such as quantum key distribution and teleportation.
In order to apply entanglement to more various and useful quantum information processing,
there are several important things which we should take into account.
One is to quantify the degree of entanglement,
and another one is to know about more properties of entanglement.
We here consider two measures of entanglement,
and investigate some properties of entanglement related to the two entanglement measures
in multipartite systems,
especially $2\otimes 2 \otimes d$ quantum systems.

Wootters' {\em concurrence}~\cite{Wootters}, $\mathcal{C}$
has been considered as one of the simplest measure of entanglement,
although there does not in general exist its explicit formula.
For any pure state $\ket{\phi}_{AB}$,
it is defined as $\mathcal{C}(\ket{\phi}_{AB})=\sqrt{2(1-\T\rho_A^2)}$,
where $\rho_A=\T_B\ket{\phi}_{AB}\bra{\phi}$.
Note that $\sqrt{2(1-\T\rho_A^2)}=2\sqrt{\det\rho_A}$ in $2\otimes d$ systems.
For any mixed state $\rho_{AB}$,
it is defined as
\begin{equation}
\mathcal{C}(\rho_{AB})=\min \sum_{k} p_k \mathcal{C}(\ket{\phi_k}_{AB}),
\label{eq:concurrence}
\end{equation}
where the minimum is taken over its all possible decompositions,
$\rho_{AB}=\sum_k p_k \ket{\phi_k}_{AB}\bra{\phi_k}$.
Recently, another measure of entanglement has been presented,
and it is called the {\em concurrence of assistance} (CoA)~\cite{LVE},
which is defined as
\begin{equation}
\mathcal{C}^a(\rho_{AB})=\max \sum_{k} p_k \mathcal{C}(\ket{\phi_k}_{AB}),
\label{eq:CoA}
\end{equation}
where the maximum is taken over all possible decompositions of $\rho_{AB}$.

In multiqubit systems,
there is an interesting property about multipartite entanglement,
which is called the {\em monogamy of entanglement} (MoE).
Coffman, Kundu, and Wootters (CKW)
first proposed the monogamy inequality~\cite{CKW}, which states
the MoE in the 3-qubit system,
\begin{equation}
\mathcal{C}_{A(BC)}^2\ge\mathcal{C}_{AB}^2+\mathcal{C}_{AC}^2,
\end{equation}
and then its generalization was proved by Osborne and Verstraete~\cite{OV}.
Symmetrically, its dual inequality in terms of the CoA for 3-qubit states,
\begin{equation}
\mathcal{C}_{A(BC)}^2\le(\mathcal{C}_{AB}^a)^2+(\mathcal{C}_{AC}^a)^2,
\end{equation}
and its generalization into $n$-qubit cases have been also shown in~\cite{GMS,GBS}.

In particular, for 3-qubit states,
it can be readily proved that the monogamy equality~\cite{LJK,YS},
\begin{equation}
\mathcal{C}_{A(BC)}^2=\mathcal{C}_{AB}^2+\left(\mathcal{C}_{AC}^a\right)^2
\label{eq:ME222}
\end{equation}
holds.
We note that this monogamy equality shows the MoE
more explicitly than the CKW inequality.
Thus, it could be important to investigate
whether the monogamy equality would be possible
in any higher dimensional tripartite quantum systems,
and could be helpful for us to understand multipartite entanglement.

In this paper, we consider the monogamy equality in $2\otimes 2 \otimes d$ systems.
We show that $\mathcal{C}_{A(BC)}=\mathcal{C}_{AB}$ if and only if $\mathcal{C}_{AC}^a=0$,
and also show that if $\mathcal{C}_{A(BC)}=\mathcal{C}_{AC}^a$ then $\mathcal{C}_{AB}=0$,
whereas there exists a state in a $2\otimes 2 \otimes d$ system such that
$\mathcal{C}_{AB}=0$ but $\mathcal{C}_{A(BC)}>\mathcal{C}_{AC}^a$.

%
%


Now, we present the first main theorem.
\begin{Thm}\label{Thm:1}
Let $\ket{\Psi}_{ABC}$ be a state in a $2\otimes 2 \otimes d$ system.
Then the followings are equivalent.
\begin{enumerate}
\item[\em (i)] $\ket{\Psi}$ is of the form
    $\ket{\phi}_A \otimes \ket{\psi}_{BC}$ or $\ket{\phi'}_C \otimes \ket{\psi'}_{AB}$.
\item[\em (ii)] $\mathcal{C}_{AC}^{a}=0$.
\item[\em (iii)] $\mathcal{C}_{A(BC)}=\mathcal{C}_{AB}$.
\end{enumerate}
\end{Thm}
In order to prove Theorem~\ref{Thm:1},
we introduce the following lemma,
which is called the {\em Lewenstein-Sanpera decomposition}
for two-qubit states~\cite{LS}.
\begin{Lem}\label{Lem:LS_decomposition}
Let $\rho$ be a density matrix on $\mathbb{C}^2\otimes \mathbb{C}^2$.
Then $\rho$ has a unique decomposition in the form
$\rho=\lambda\rho_s+(1-\lambda)P_e$,
where $\rho_s$ is a separable density matrix,
$P_e=\ket{\Psi_e}\bra{\Psi_e}$ for a pure entangled state $\ket{\Psi_e}$,
and $\lambda\in [0,1]$ is maximal.
\end{Lem}

We now give the proof of the first main theorem.

\begin{proof}[Proof of Theorem~\ref{Thm:1}]
We first prove that (i) is equivalent to (ii).
Since $\rho_{AC}$ is in the form of
$\ket{\psi}_A\bra{\psi}\otimes\sigma_C$ or $\ket{\psi}_C\bra{\psi}\otimes\sigma_A$,
it is trivial that $\mathcal{C}_{AC}^{a}=0$.
Conversely, suppose that
$\rho_{AC}$ is not in the form of
$\ket{\psi}_A\bra{\psi}\otimes\sigma_C$ or $\ket{\psi}_C\bra{\psi}\otimes\sigma_A$.
Then $\rho_A$ and $\rho_C$ have at least rank 2.
Since $\mathcal{C}_{AC}^a = 0$,
\begin{equation}
\rho_{AC}=\sum_i {p_i}\ket{\phi_i}_A\bra{\phi_i}\otimes\ket{\psi_i}_C\bra{\psi_i},
\label{eq:rho_AC}
\end{equation}
and there exists at least one pair $(i,j)$ such that
$|\inn{\phi_i}{\phi_j}|\neq 1$ and $|\inn{\psi_i}{\psi_j}|\neq 1$.
By Hughston-Jozsa-Wootters (HJW) theorem~\cite{HJW},
$\rho_{AC}=\sum_k{q_k}\ket{\Phi_k}_{AC}\bra{\Phi_k}$ such that at least one
\begin{equation}
\ket{\Phi_k}_{AC}=\alpha\ket{\phi_i}_A\ket{\psi_i}_C
+\beta\ket{\phi_j}_A\ket{\psi_j}_C
\label{eq:Phi_k}
\end{equation}
is entangled ($\alpha\neq 0$ and $\beta\neq 0$),
and hence $\mathcal{C}_{AC}^{a} > 0$.

Since if $\ket{\Psi}_{ABC}$ is in the form
$\ket{\phi}_A \otimes \ket{\psi}_{BC}$ or $\ket{\phi'}_C \otimes \ket{\psi'}_{AB}$
then it is clear that $\mathcal{C}_{A(BC)}=\mathcal{C}_{AB}$,
the final one for completing the proof of this theorem,
is to show that if $\mathcal{C}_{A(BC)}=\mathcal{C}_{AB}$
then $\ket{\Psi}_{ABC}$ has the form
$\ket{\phi}_A \otimes \ket{\psi}_{BC}$ or $\ket{\phi'}_C \otimes \ket{\psi'}_{AB}$.

We assume that $\mathcal{C}_{A(BC)}=\mathcal{C}_{AB}\neq 0$
(If $\mathcal{C}_{A(BC)} = \mathcal{C}_{AB} = 0$ then
$\ket{\Psi}_{ABC}$ is of the form $\ket{\phi}_A \otimes \ket{\psi}_{BC}$,
and so this theorem is trivially true).
Then we clearly have $\mathcal{C}_{A(BC)}=\mathcal{C}_{AB}^a=\mathcal{C}_{AB}$,
that is, the average concurrence of any decomposition of $\rho_{AB}$ is equal to
$\mathcal{C}_{A(BC)}$.
By Lemma~\ref{Lem:LS_decomposition},
$\rho_{AB}=\lambda\rho_s + (1-\lambda)P_e$,
where $\rho_s$ is separable and $P_e$ is purely entangled.
Then since $\mathcal{C}_{A(BC)}=\mathcal{C}_{AB}=(1-\lambda)\mathcal{C}_{AB}(P_e)$,
we can see that $\rho_s$ is in the form of $\ket{0}_A\bra{0}\otimes \sigma_B$ or
$(x\ket{0}_A\bra{0}+y\ket{1}_A\bra{1})\otimes \ket{0}_B\bra{0}$
up to local unitary operations.

We now let
\begin{equation}
\begin{pmatrix}
 a & b \\
 b^{\ast} & c \\
\end{pmatrix}
\equiv (1-\lambda)\T_B(P_e).
\label{eq:matrix01}
\end{equation}
Then $(1-\lambda)\mathcal{C}_{AB}(P_e)=2\sqrt{ac-|b|^2}$,
and
\begin{equation}
\rho_A=
\begin{pmatrix}
\lambda+a& b\\
b^*& c
\end{pmatrix}
\hbox{ or }
\begin{pmatrix}
\lambda{x}+a& b\\
b^*& \lambda{y}+c
\end{pmatrix}.
\label{eq:matrix02}
\end{equation}
Thus, since $\mathcal{C}_{A(BC)}=2\sqrt{(\lambda+a)c-|b|^2}$
or $2\sqrt{(\lambda{x}+a)(\lambda{y}+c)-|b|^2}$,
it is obtained that $\lambda=0$, that is, $\rho_{AB}=P_e$.
Therefore, we can conclude that $\ket{\Psi}_{ABC}$ is of the form
$\ket{\phi}_C \otimes \ket{\psi}_{AB}$.
\end{proof}

We now present the second main theorem.
\begin{Thm}\label{Thm:2}
If $\mathcal{C}_{A(BC)}=\mathcal{C}_{AC}^{a}$ then $\mathcal{C}_{AB}=0$.
\end{Thm}

For the proof of Theorem~\ref{Thm:2},
we introduce the two following lemmas.
One is as follows.
\begin{Lem}\label{Lem:rank}
If $\rho$ and $\sigma$ are $2\times 2$ positive matrices with
$\mathrm{rank}(\rho) =1$ and $\mathrm{rank}(\sigma)=2$, respectively
then
for any $\lambda_j\ge 0$,
$\sqrt{\det(\lambda_0 \rho + \lambda_1 \sigma)}\ge \lambda_1  \sqrt{\det\sigma}$,
where the equality holds if and only if $\lambda_{0}=0$ or $\lambda_{1}=0$.
If $\rho$ and $\sigma_j$ are $2\times 2$ positive matrices with
$\mathrm{rank}(\rho) =1$ and $\mathrm{rank}(\sigma_j)=2$, respectively
then
for any $\alpha, \beta_j \ge 0$,
\begin{equation}
\sqrt{\det(\alpha \rho + \beta_0 \sigma_0 + \beta_1 \sigma_1)} \geq
\sqrt{\det(\beta_0\sigma_0+\beta_1\sigma_1)},
\label{eq:Lem2}
\end{equation}
where the equality holds if and only if $\alpha=0$ or $\beta_j=0$.
\end{Lem}

\begin{proof}
To begin with, we show the first statement.
Without loss of generality, we may assume that
$\rho=\ket{\psi}\bra{\psi}$ and $\sigma=a \ket{0}\bra{0}+ b \ket{1}\bra{1}$,
where $\ket{\psi}=x\ket{0}+y\ket{1}$ and $a, b >0$.
Then
\begin{eqnarray}
\sqrt{\det(\lambda_0 \rho + \lambda_1 \sigma)}
&=& \sqrt{\lambda_{0}\lambda_{1}(|x|^{2}b+|y|^{2}a)+\lambda_{1}^{2}ab}
\nonumber\\
&\ge& \sqrt{\lambda_{1}^{2}ab} \nonumber \\
&=&\lambda_1  \sqrt{\det\sigma}.
\label{eq:rank12}
\end{eqnarray}
It is clear that
the equality in (\ref{eq:rank12}) holds if and only if $\lambda_{0}=0$ or $\lambda_{1}=0$.
Similarly, we can show the second statement.
\end{proof}

The other lemma is called the {\em Minkowski determinant inequality theorem}~\cite{HJ}.
\begin{Lem}\label{Lem:MDI}
If $n\times n$ matrices $A$, $B$ are positive definite, then
\begin{equation}
\left[\det(A+B)\right]^{1/n}\ge\left(\det A\right)^{1/n}+\left(\det B\right)^{1/n}.
\label{eq:MDI}
\end{equation}
The equality in (\ref{eq:MDI}) holds if and only if $B=cA$ for some $c\ge 0$.
\end{Lem}

In the proof of the second main theorem,
we will use Lemma~\ref{Lem:MDI}
just in the case of $n=2$.

\begin{proof}[Proof of Theorem~\ref{Thm:2}]
We first let
\begin{equation}
\rho_{AC}=\sum_{i\in I} \lambda_i \ket{\psi_i}_{AC}\bra{\psi_i}
\label{eq:optimal}
\end{equation}
be an optimal decomposition of $\rho_{AC}$ for the CoA, $\mathcal{C}_{AC}^a$.
Then we can consider the three cases
according to the rank of $\T_{C}(\ket{\psi_i}_{AC}\bra{\psi_i})$;
(i)~$\mathrm{rank} \left[\T_{C}(\ket{\psi_i}_{AC}\bra{\psi_i})\right]=1$ for all $i\in I$,
(ii)~there exist two nonempty subsets $I_1$ and $I_2=I-I_1$ of $I$ such that
$\mathrm{rank} \left[\T_{C}(\ket{\psi_i}_{AC}\bra{\psi_i})\right]=1$ for all $i\in I_1$
and $\mathrm{rank} \left[\T_{C}(\ket{\psi_j}_{AC}\bra{\psi_j})\right]=2$ for all $j\in I_2$,
(iii)~$\mathrm{rank} \left[\T_{C}(\ket{\psi_i}_{AC}\bra{\psi_i})\right]=2$ for all $i\in I$.

We now prove this theorem case by case.

(Case i) Assume that $\mathrm{rank} \left[\T_{C}(\ket{\psi_i}_{AC}\bra{\psi_i})\right]=1$
for all $i\in I$.
Then since $\ket{\psi_i}_{AC}\bra{\psi_i}$'s are all pure and separable states,
$\mathcal{C}_{AC}^{a}=0$ and so $\mathcal{C}_{A(BC)}=0$.
Thus, we have $\ket{\Psi}_{ABC}$ is of the form $\ket{\phi}_A \otimes \ket{\psi}_{BC}$,
and it immediately follows that $\mathcal{C}_{AB}=0$.

(Case ii) Assume that
there exist two nonempty subsets $I_1$ and $I_2=I-I_1$ of $I$ such that
$\mathrm{rank} \left[\T_{C}(\ket{\psi_i}_{AC}\bra{\psi_i})\right]=1$ for all $i\in I_1$
and $\mathrm{rank} \left[\T_{C}(\ket{\psi_j}_{AC}\bra{\psi_j})\right]=2$ for all $j\in I_2$.
The by Lemma~\ref{Lem:rank} and Lemma~\ref{Lem:MDI},
we can obtain the following inequality.
\begin{widetext}
\begin{eqnarray}
\mathcal{C}_{A(BC)}
&=&
2\sqrt{\det\left[\sum_{i\in I_1}\lambda_i \T_C(\ket{\psi_i}_{AC}\bra{\psi_i})
+\sum_{j\in I_2}\lambda_j \T_C(\ket{\psi_j}_{AC}\bra{\psi_j})\right]}
\ge
2\sqrt{\det\left[\sum_{j\in I_2}\lambda_j \T_C(\ket{\psi_j}_{AC}\bra{\psi_j})\right]}
\nonumber \\
&\ge& 2\sum_{j\in I_2}\lambda_j\sqrt{\det(\T_C(\ket{\psi_j}_{AC}\bra{\psi_j}))}
= \mathcal{C}_{AC}^{a}.
\label{eq:rank1122}
\end{eqnarray}
\end{widetext}
Since $\mathcal{C}_{A(BC)}=\mathcal{C}_{AC}^{a}$,
the equality in the first inequality should hold,
and hence $\lambda_i=0$ for all $i\in I_1$
or $\lambda_j=0$ for all $j\in I_2$ by Lemma~\ref{Lem:rank}.
This means that it is sufficient to consider the cases (i) and (iii).

(Case iii) We assume that
$\mathrm{rank} \left[\T_{C}(\ket{\psi_i}_{AC}\bra{\psi_i})\right]=2$ for all $i\in I$.
Then by Lemma~\ref{Lem:MDI},
\begin{eqnarray}
\mathcal{C}_{A(BC)}
&=& 2\sqrt{\det\left[\T_{C}\sum_{i\in I}\lambda_{i}\ket{\psi_i}_{AC}\bra{\psi_i}\right]}
\nonumber \\
&\ge& 2\sum_{i\in I}\lambda_{i}\sqrt{\det\left(\T_{C}\ket{\psi_i}_{AC}\bra{\psi_i}\right)}
\nonumber \\
&=& \mathcal{C}_{AC}^{a},
\label{eq:ineq2}
\end{eqnarray}
and the equality in the inequality (\ref{eq:ineq2}) holds if and only if
$\T_{C}\ket{\psi_i}_{AC}\bra{\psi_i}=\rho_{A}$
for all $i\in I$.

Let $\rho_A=\mu_{0}\ket{0}_A\bra{0}+\mu_{1}\ket{1}_A\bra{1}$
be its spectral decomposition.
By the Gisin-Hughston-Jozsa-Wootters theorem~\cite{HJW,Gisin}, for $0, 1\in I$,
there is a unitary operator $U$ such that
\begin{eqnarray}
\ket{\psi_{0}}_{AC}&=&\sqrt{\mu_{0}}\ket{0}_{A}\ket{0}_{C}+\sqrt{\mu_{1}}\ket{1}_{A}\ket{1}_{C},
\nonumber \\
\ket{\psi_{1}}_{AC}&=&\sqrt{\mu_{0}}\ket{0}_{A}U\ket{0}_{C}+\sqrt{\mu_{1}}\ket{1}_{A}U\ket{1}_{C}.
\label{eq:Unitary}
\end{eqnarray}
Let
$\rho_{AC}=\nu_0\ket{\phi_0}_{AC}\bra{\phi_0}+\nu_1\ket{\phi_1}_{AC}\bra{\phi_1}$
be the spectral decomposition of $\rho_{AC}$.
Then since $\mathrm{rank} (\rho_{AC})=2$,
the eigenvectors $\ket{\tilde{\phi_{0}}}=\sqrt{\nu_0}\ket{\phi_0}$
and $\ket{\tilde{\phi_{1}}}=\sqrt{\nu_1}\ket{\phi_1}$
should be linear combinations of $\ket{\psi_{0}}$ and $\ket{\psi_{1}}$.
It follows that
\begin{equation}
\ket{\Psi}_{ABC}=\ket{\tilde{\phi_{0}}}_{AC}\ket{0}_B+\ket{\tilde{\phi_{1}}}_{AC}\ket{1}_B,
\label{eq:psi}
\end{equation}
where $\ket{\tilde{\phi_{0}}}=x_0\ket{\psi_{0}}+x_1\ket{\psi_{1}}$
and $\ket{\tilde{\phi_{1}}}=y_0\ket{\psi_{0}}+y_1\ket{\psi_{1}}$.
Let
\begin{equation}
\ket{\Psi'}_{ABC}=\ket{\tilde{\phi_{0}}'}_{AC}\ket{0}_B+\ket{\tilde{\phi_{1}}'}_{AC}\ket{1}_B,
\label{eq:psi1}
\end{equation}
where $\ket{\tilde{\phi_{0}}'}=x_1^*\ket{\psi_{0}}+x_0^*\ket{\psi_{1}}$ and
$\ket{\tilde{\phi_{1}}'}=y_1^*\ket{\psi_{0}}+y_0^*\ket{\psi_{1}}$.
Then, by tedious but straightforward calculations,
we can check that
the partial transpose $\rho_{AB}^{T_B}$ of $\rho_{AB}$ is equal to
$\rho'_{AB}=\T_C(\ket{\Psi'}_{ABC}\bra{\Psi'})$,
and thus $\rho_{AB}$ has positive partial transposition (PPT).
Therefore, $\mathcal{C}_{AB}=0$.
\end{proof}

So far, we have seen the case that the monogamy equality holds in $2\otimes 2\otimes d$ systems.
We now exhibit a counterexample that the monogamy equality does not hold,
in particular, $\mathcal{C}_{AB}=0$ but $\mathcal{C}_{A(BC)} > \mathcal{C}_{AC}^{a}$.
\begin{Exam}
Consider two orthogonal states in the $2\otimes 3$ quantum system,
$\ket{x}=(\ket{02}+\sqrt{2}\ket{10})/{\sqrt{3}}$,
$\ket{y}=(\ket{12}+\sqrt{2}\ket{01})/{\sqrt{3}}$.
We now take into account the following state in the $2\otimes 2\otimes 3$ quantum system,
\begin{eqnarray}
\ket{\Psi}_{ABC}&=& \frac{1}{\sqrt{2}}\ket{x}_{AC}\ket{0}_B
+\frac{1}{\sqrt{2}}\ket{y}_{AC}\ket{1}_B
\nonumber \\
&=&\frac{1}{\sqrt{6}}\ket{002}_{ABC}+\frac{1}{\sqrt{3}}\ket{100}_{ABC}
\nonumber \\
&&+\frac{1}{\sqrt{6}}\ket{112}_{ABC}+\frac{1}{\sqrt{3}}\ket{011}_{ABC}.
\label{eq:ex_Psi}
\end{eqnarray}
Then since $\rho_{A}=(\ket{0}_{A}\bra{0}+\ket{1}_{A}\bra{1})/2$,
it is clear that $\mathcal{C}_{A(BC)}=1$,
and since $\rho_{AC}=(\ket{x}_{AC}\bra{x}+\ket{y}_{AC}\bra{y})/2$,
by the HJW theorem, for any decompositions $\rho_{AC}=\sum_{i}p_{i}\ket{\phi_{i}}_{AC}\bra{\phi_{i}}$,
$\sqrt{p_{i}}\ket{\phi_{i}}_{AC}=(c_{i1}\ket{x}_{AC}+c_{i2}\ket{y}_{AC})/\sqrt{2}$
for some unitary operator $(c_{ij})$ with $2p_{i}=|c_{i1}|^{2}+|c_{i2}|^{2}$.
Then
\begin{equation}
2p_{i}\T_{C}(\ket{\phi_{i}}_{AC}\bra{\phi_{i}})
=\frac{1}{3}
\begin{pmatrix}
  |c_{i1}|^{2}+2|c_{i2}|^{2} & c_{i1}c_{i2}^{*} \\
  c_{i2}c_{i1}^{*} & |c_{i2}|^{2}+2|c_{i1}|^{2}
\end{pmatrix},
\label{eq:decomp_matrix}
\end{equation}
and hence
\begin{equation}
\T_{C}(\ket{\phi_{i}}_{AC}\bra{\phi_{i}})
=\frac{1}{3}I_{A}+\frac{1}{3}\ket{\psi_{i}}_{A}\bra{\psi_{i}}
\label{eq:ex_trC}
\end{equation}
with $\ket{\psi_{i}}=(c_{i2}^*\ket{0}+c_{i1}^*\ket{1})/{\sqrt{2p_{i}}}$.
Thus we obtain that $\mathcal{C}_{AC}=\frac{2\sqrt{2}}{3}=\mathcal{C}_{AC}^{a}$.
Since
\begin{equation}
\rho_{AB}=\frac{1}{6}
\begin{pmatrix}
  1 & 0 & 0 & 1 \\
  0 & 2 & 0 & 0 \\
  0 & 0 & 2 & 0 \\
  1 & 0 & 0 & 1
\end{pmatrix}
\label{eq:ex_rhoAB}
\end{equation}
clearly has PPT,
$\mathcal{C}_{AB}=0$.
Therefore,
there exists a quantum state in the $2\otimes 2 \otimes 3$ system
such that $\mathcal{C}_{AB}=0$,
but $\mathcal{C}_{A(BC)}=1 > \frac{2\sqrt{2}}{3}=\mathcal{C}_{AC}^{a}$.
\end{Exam}


In conclusion,
we have considered the monogamy equality
in $2\otimes 2 \otimes d$ quantum systems.
We have shown that $\mathcal{C}_{A(BC)}=\mathcal{C}_{AB}$ if and only if $\mathcal{C}_{AC}^a=0$,
and have also shown that
if $\mathcal{C}_{A(BC)}=\mathcal{C}_{AC}^a$ then $\mathcal{C}_{AB}=0$,
while there exists a state in a $2\otimes 2 \otimes d$ system such that
$\mathcal{C}_{AB}=0$ but $\mathcal{C}_{A(BC)}>\mathcal{C}_{AC}^a$.
However, in $2\otimes 2 \otimes d$ quantum systems,
the monogamy inequality in terms of the concurrence and the CoA,
$\mathcal{C}_{A(BC)}^2\ge\mathcal{C}_{AB}^2+\left(\mathcal{C}_{AC}^a\right)^2$,
has been still unknown.

D.P.C. was supported by the Korea Science and Engineering Foundation
(KOSEF) grant funded by the Korea government (MOST) (No.~R01-2006-000-10698-0),
J.S.K was supported by Alberta's informatics Circle of Research Excellence (iCORE), and
S.L. was supported by the Korea Research Foundation Grant funded by the Korean Government
(MOEHRD, Basic Research Promotion Fund) (KRF-2007-331-C00049).



\begin{thebibliography}{1}

%
%
\bibitem{Wootters} W.K.~Wootters,
Phys. Rev. Lett. {\bf 80}, 2245 (1998).
%
\bibitem{LVE}
T.~Laustsen, F.~Verstraete, and S.J.~van~Enk,
Quantum Inf. Comput. {\bf 3}, 64 (2003).
%
\bibitem{CKW}
V.~Coffman, J.~Kundu, and W.K.~Wootters,
Phys. Rev. A {\bf 61}, 052306 (2000).
%
\bibitem{OV}
T.J.~Osborne and F.~Verstraete,
Phys. Rev. Lett. {\bf 96}, 220503 (2006).
%
\bibitem{GMS}
G.~Gour, D.A.~Meyer, and B.C.~Sanders,
Phys. Rev. A {\bf 72}, 042329 (2005).
%
\bibitem{GBS}
G.~Gour, S.~Bandyopadhyay, and B.C.~Sanders,
J. Math. Phys. {\bf 48}, 012108 (2007).
%
\bibitem{LJK}
S.~Lee, J.~Joo, and J.~Kim,
Phys. Rev. A {\bf 72}, 024302 (2005).
%
\bibitem{YS}
C-s.~Yu and H-s.~Song,
Phys. Rev. A {\bf 76}, 022324 (2007).
%
\bibitem{LS}
M.~Lewenstein and A.~Sanpera,
Phys. Rev. Lett. {\bf 80}, 2261 (1997).
%
\bibitem{HJW} L.P.~Hughston, R.~Jozsa, and W.K.~Wootters,
Phys. Lett. A {\bf 183}, 14 (1993).
%
\bibitem{HJ} R.A.~Horn and C.R. Johnson,
{\em Matrix Analysis},
(Cambridge University Press, New York, 1985).
%
\bibitem{Gisin} N.~Gisin,
Helv. Phys. Acta {\bf 62}, 363 (1989).
%

\end{thebibliography}
\end{document}